\documentclass{article}
\usepackage{graphicx}
\usepackage{amssymb,amsmath,mathrsfs,amsthm}
\usepackage{subfigure}
\usepackage{xspace}
\usepackage{enumerate}

\newtheorem{theorem}{Theorem}
\newtheorem{observation}[theorem]{Observation}
\newtheorem{definition}[theorem]{Definition}
\newtheorem{proposition}[theorem]{Proposition}
\newtheorem{lemma}[theorem]{Lemma}

 \newcommand{\HHH}{\mathcal{H}}

 \newcommand{\PPP}{\mathcal{P}}

\newcommand{\varsubscript}[2]{{#1}_{\mbox{\footnotesize #2}}}
\newcommand{\BC}{{\mathcal B}^+}
\newcommand{\BT}{{\mathcal B}^\div}
\newcommand{\portaledges}{\varsubscript{E}{portal}}

\newcommand{\MGthin}{\varsubscript{\MG}{thin}}
\newcommand{\preprocess}{\mbox{\sc Preprocess}}
\newcommand{\planarize}{\mbox{\sc Planarize}}
\newcommand{\thinning}{\mbox{\sc Thinning}}

\newcommand{\longvar}[1]{\mathop{\mathrm{#1}}\nolimits}
\newcommand{\OPT}{\longvar{OPT}}

\newcommand{\MG}{\longvar{MG}}
\newcommand{\CG}{\longvar{CG}}
\newcommand{\SPT}{\longvar{SPT}}

\newcommand{\tw}{\operatorname{tw}}

\newcommand{\dist}{\operatorname{dist}}

\newcommand{\vrt}[1]{V(#1)}

\newcommand{\vrtG}{\vrt{G}}

\newcommand{\edge}[1]{E(#1)}

\newcommand{\edgeG}{\edge{G}}

\newcommand{\edgeT}{\edge{T}}

\newcommand{\Gthin}{\varsubscript{G}{thin}}
\newcommand{\Gspan}{\varsubscript{G}{span}}

\newcommand{\Oof}{\mathcal{O}}

\newcommand{\set}[1]{\ensuremath{\{#1\}}}
\newcommand{\length}[1]{\ensuremath{\ell(#1)}}

\renewcommand{\epsilon}{\varepsilon}
\renewcommand{\emptyset}{\varnothing} 

\newcommand{\CompClass}[1]{\ensuremath{\mathsf{#1}}\xspace}
\newcommand{\NP}{\CompClass{NP}}

\newcommand{\APX}{\CompClass{APX}}

\newcommand{\PTAS}{\CompClass{PTAS}}

\newcommand{\poly}{\ensuremath{\mathrm{poly}}}

\newcommand{\myproblemname}[1]{\ensuremath{\mbox{\sc #1}}\xspace}
\newcommand{\steiner}{\myproblemname{Steiner Tree}}
\newcommand{\steinerforest}{\myproblemname{Steiner Forest}}

\newcommand{\subtsp}{\myproblemname{Subset Tsp}}
\newcommand{\survive}{\myproblemname{Survivable Network}}

\begin{document}

\title{Polynomial-time approximation schemes for subset-connectivity
    problems in bounded-genus graphs}


\author{Glencora Borradaile \and
        Erik D.\ Demaine \and
        Siamak Tazari
}





\maketitle

\begin{abstract}
  We present the first polynomial-time approximation schemes (PTASes)
  for the following subset-connectivity problems in edge-weighted
  graphs of bounded genus: Steiner tree, low-connectivity
  survivable-network design, and subset TSP.  The schemes run in $\Oof(n
  \log n)$ time for graphs embedded on both orientable and
  nonorientable surfaces.  This work generalizes the PTAS frameworks
  of Borradaile, Klein, and Mathieu (2007) from
  planar graphs to bounded-genus graphs: any future problems shown to
  admit the required structure theorem for planar graphs will
  similarly extend to bounded-genus graphs.
\end{abstract}


\section{Introduction}

In many practical scenarios of network design, input graphs
have a natural drawing on the sphere or equivalently the plane.
In most cases, these embeddings have few crossings, either to avoid digging
multiple levels of tunnels for fiber or cable or to avoid building overpasses
in road networks.  But a few crossings are common, and can easily come in
bunches where one tunnel or overpass might carry several links or roads.
Thus we naturally arrive at graphs of small (bounded) genus,
which is the topic of this work.

We develop a $\PTAS$ framework for subset-connectivity problems on
edge-weighted graphs of bounded genus.  In general, we are given a
subset of the nodes, called \emph{terminals}, and the goal is to
connect the terminals together with some substructure of the graph by
using cost within $1+\epsilon$ of the minimum possible cost.  Our
framework applies to three well-studied problems in this framework.
In \steiner, the substructure must be connected, and thus forms a
tree.  In \subtsp, the substructure must be a cycle; to guarantee
existence, the cycle may traverse vertices and edges multiple times,
but pays for each traversal.  In $\{0,1,2\}$-edge-connectivity
\survive, the substructure must have $\min\{c_x, c_y\}$ edge-disjoint
paths connecting vertices $x$ and~$y$, where each $c_x \in
\{0,1,2\}$; we allow the substructure to include multiple copies of an
edge in the graph, but pay for each copy.  In particular, if $c_x = 1$
for all terminals $x$ and~$y$, then we obtain the Steiner tree
problem; if $c_x = 2$ for all terminals $x$ and~$y$, then we obtain
the minimum-cost $2$-edge-connected multi-subgraph problem.

Our framework yields the first $\PTAS$ for all of these problems in
bounded-genus graphs.  These $\PTAS$es are efficient, running in
$\Oof(f(\epsilon,g) \, n + h(g) \, n \log n) = \Oof_{\epsilon,g}(n
\log n)$ time for graphs embedded on orientable surfaces and
nonorientable surfaces.  (We usually omit the mention of
$f(\epsilon,g)$ and $h(g)$ by assuming $\epsilon$ and $g$ are
constant, but we later bound $f(\epsilon,g)$ as singly exponential in
a polynomial in~$1/\epsilon$ and $g$ and $h(g)$ as singly exponential
in $g$.)  In contrast, the problems we consider are $\APX$-complete
(and constant-factor-approximable) for general graphs.

We build upon the recent $\PTAS$ framework of Borradaile, Klein, and
Mathieu~\cite{BorradaileKM09} for subset-connectivity problems on
planar graphs.  In fact, our result is strictly more general: any
problem to which the previous planar-graph framework applies
automatically works in our framework as well, resulting in a $\PTAS$ for
bounded-genus graphs.  For example, Borradaile and
Klein~\cite{BorradaileKlein08} have recently given a $\PTAS$ for the
$\{0,1,2\}$-edge-connectivity \survive problem
using the planar framework.  This will imply a similar result in
bounded genus graphs.  In contrast to the planar-graph framework, our
$\PTAS$es have the attractive feature that they run correctly on all
graphs with the performance degrading with the genus.

Our techniques for attacking bounded-genus graphs include two recent
results: decompositions into bounded-treewidth graphs via
contractions~\cite{DemaineHM07} and fast algorithms for finding the
shortest noncontractible cycle~\cite{CabelloChambers07}.  We also use
a simplified version of an algorithm for finding a short sequence of
loops on a topological surface~\cite{EricksonWhittlesey05}, and
sophisticated dynamic programming. Our aim is to prove the following theorem:

\begin{theorem}\label{thm:main_ptas} 
  There exists a $\PTAS$ for the \steiner, \subtsp, and
  $\set{0,1,2}$-edge-connected \survive problems in edge-weighted graphs of
  genus $g$ with running time $\Oof(2^{\poly(\epsilon^{-1},g)}n +
  2^{\poly(g)}n \log n)$.
\end{theorem}


\section{Preliminaries}

All graphs $G = (V,E)$ have $n$ vertices, $m$ edges and are undirected
with edge lengths (weights).  The length of an edge $e$, subgraph $H$, and set
of subgraphs $\HHH$ are denoted $\length{e}$, $\length{H}$ and
$\length{{\HHH}}$, respectively.  The shortest distance between
vertices $x$ and $y$ in a graph $G$ is denoted $\dist_G(x,y)$.  The
boundary of a graph $G$ embedded in the plane is denoted by $\partial
G$. For an edge $e=uv$, we define the operation of \emph{contracting}
$e$ as identifying $u$ and $v$ and removing all loops and duplicate edges.

We use the basic terminology for embeddings as outlined
in~\cite{MoharThomassen01}. In this paper, an embedding refers to a
\emph{2-cell embedding}, i.e. a drawing of the vertices and faces of
the graph as points and arcs on a surface such that every face is
homeomorphic to an open disc. Such an embedding can be described
purely combinatorially by specifying a \emph{rotation system}, for the
cyclic ordering of edges around vertices of the graph, and a
\emph{signature} for each edge of the graph; we use this notion of a
\emph{combinatorial embedding}. A combinatorial embedding of a graph
$G$ naturally induces such a 2-cell embedding on each subgraph of
$G$. We only consider compact surfaces without boundary. When we refer
to a planar embedding, we actually mean an embedding in the 2-sphere.
If a surface contains a subset homeomorphic to a M\"obius strip, it is
\emph{nonorientable}; otherwise it is \emph{orientable}. For a 2-cell
embedded graph $G$ with $f$ facial walks, the number $g = 2 + m - n -
f$ is called the Euler genus of the surface. The Euler genus is equal
to twice the usual genus for orientable surfaces and equals the usual
genus for nonorientable surfaces. The \emph{dual} of an embedded
graph $G$ is defined as having the set of faces of $G$ as its vertex
set and having an edge between two vertices if the corresponding faces
of $G$ are adjacent.  We denote the dual graph by $G^\star$ and
identify each edge of $G$ with its corresponding edge in $G^\star$.  A
cycle of an embedded graph is \emph{contractible} if it can be
continuously deformed to a point; otherwise it is
\emph{noncontractible}. The operation of \emph{cutting along a
  2-sided cycle} $C$ is essentially: partition the edges adjacent to
$C$ into left and right edges and replace $C$ with two copies $C_\ell$
and $C_r$, adjacent to the left or right edges, accordingly. The
inside of these new cycles is ``patched'' with two new faces. If the
resulting graph is disconnected, the cycle is called
\emph{separating}, otherwise \emph{nonseparating}.  Cutting along a
1-sided cycle $C$ on nonorientable surfaces is defined similarly,
only that $C$ is replaced by one bigger cycle $C'$ that contains every
edge of $C$ exactly twice.  See~\cite[pages
  105--106]{MoharThomassen01} for further technical details.

Next we define the notions related to treewidth as introduced by
Robertson and Seymour~\cite{GraphMinors02}. A \emph{tree
  decomposition} of a graph $G$ is a pair $(T,\chi)$, where
$T=(I,F)$ is a tree and $\chi=\set{\chi_i | i \in I}$ is a family of
subsets of $V(G)$, called \emph{bags}, such that
\begin{enumerate} 
\item every vertex of $G$ appears in some bag of $\chi$;
\item for every edge $e=uv$ of $G$, there exists a bag that contains
  both $u$ and $v$;
\item for every vertex $v$ of $G$, the set of bags that contain $v$
  form a connected subtree $T_v$ of $T$.
\end{enumerate}
The \emph{width} of a tree decomposition is the maximum size of a bag
in $\chi$ minus $1$. The \emph{treewidth} of a graph $G$, denoted by
$\tw(G)$, is the minimum width over all possible tree decompositions
of $G$. 

The input graph is $G_0=(V_0,E_0)$ and has genus $g_0$; the
terminal set is $Q$. We assume $G_0$ is equipped with a combinatorial
embedding; such an embedding can be found in linear time, if the genus
is known to be fixed, see~\cite{Mohar99}. Let $\PPP$ be the considered
subset-connectivity problem. In Section~\ref{sec:preprocess}, we show
how to find a subgraph $G=(V,E)$ of $G_0$, so that for $0 \leq
\epsilon \leq 1$ any $(1+\epsilon)$-approximate solution of $\PPP$ in
$G_0$ also exists in $G$. Hence, we may use $G$ instead of $G_0$
in the rest of the paper. Note that as a subgraph of $G_0$, $G$ is
automatically equipped with a combinatorial embedding.

Let $\OPT$ denote the length of a optimal Steiner tree spanning
terminals~$Q$. We define $\OPT_{\PPP}$ to be the length of an
optimal solution to problem $\PPP$.  For the problems that we
solve, we require that $\OPT_{\PPP} = \Theta(\OPT)$ and in
particular that $\OPT \leq \OPT_{\PPP} \leq \mu \OPT$.  The
constant $\mu$ will be used in Section~\ref{sec:genus-mg} and is equal
to $2$ for both the subset TSP and $\{0,1,2\}$-edge-connectivity
problems.  This requirement is also needed for the planar case;
see~\cite{BorradaileKlein08}.  Because $\OPT_{\PPP} \geq \OPT$, upper bounds
in terms of $\OPT$ hold for all the problems herein.  As a result, we
can safely drop the $\PPP$ subscript throughout the paper.

We show how to obtain a $(1+c\epsilon)\OPT_{\PPP}$ solution for
a fixed constant $c$.  To obtain a $(1+\epsilon)\OPT_{\PPP}$
solution, we can simply use $\epsilon' = \epsilon/c$ as input to the
algorithm.


\section{Mortar Graph and Structure Theorem}
\label{sec:mg}

In~\cite{BorradaileKM09}, Borradaile, Klein
and Mathieu developed a $\PTAS$ for the Steiner tree problem in planar
graphs.  The method involves finding a grid-like subgraph called the
{\em mortar graph} that spans the input terminals and has length
$\Oof(\OPT)$.  The set of feasible Steiner trees is restricted to those
that cross between adjacent faces of the mortar graph only at a small
number (per face of the mortar graph) of pre-designated vertices
called {\em portals}.  A Structure Theorem guarantees the existence of
a nearly optimal solution (one that has length at most
$(1+\epsilon)\OPT$) in this set.  We review the details that are
relevant to this work and generalize to genus-$g$ graphs.

Here we define the mortar graph in such a way that generalizes to
higher genus graphs.  A path $P$ in a graph $G$ is {\em
  $\epsilon$-short in $G$} if for every pair of vertices $x$ and $y$
on $P$, the distance from $x$ to $y$ along $P$ is at most
$(1+\epsilon)$ times the distance from $x$ to $y$ in $G$:
$\dist_P(x,y) \leq (1+\epsilon)\dist_G(x,y)$.  Given a graph $G$
embedded on a surface and a set of terminals $Q$, a mortar graph is a
subgraph of $G$ with the following properties:

\begin{definition}[Mortar Graph and Bricks]\label{def:mg}
  Given a graph $G$ embedded on a surface of genus $g$, a set of
  terminals $Q$, and a number $0 < \epsilon \leq 1$, consider a
  subgraph $\MG := \MG(G,Q,\epsilon)$ of $G$ spanning $Q$ such that
  each facial walk of $\MG$ encloses an area homeomorphic to an open
  disk. For each face $F$ of $\MG$, we construct a \emph{brick} $B$ of
  $G$ by cutting $G$ along the facial walk $\partial F$; $B$ is the
  subgraph of $G$ embedded inside the face, including $\partial F$. We
  denote this facial walk as the \emph{mortar boundary} $\partial B$
  of $B$. We define the \emph{interior} of $B$ as $B$ without the
  edges of $\partial B$. We call $\MG$ a \emph{mortar graph} if for
  some constants $\alpha(\epsilon,g)$ and $\kappa(\epsilon,g)$ (to be
  defined later), we have $\ell(\MG) \leq \alpha \OPT$ and every brick $B$
  satisfies the following properties:
  \begin{enumerate} 
  \item $B$ is planar.
  \item The boundary of $B$ is the union of four paths in the
    clockwise order
    $W$, $N$, $E$, $S$.
  \item Every terminal of $Q$ that is in $B$ is on $N$ or on $S$.
  \item $N$ is 0-short in $B$, and every proper subpath of $S$ is
    $\epsilon$-short in $B$.
  \item There exists a number $k \leq \kappa$ and vertices $s_0, s_1,
    s_2, \ldots, s_k$ ordered from left to right along $S$ such
    that, for any vertex $x$ of $S[s_i,s_{i+1})$, the distance from
    $x$ to $s_i$ along $S$ is less than $\epsilon$ times the
    distance from $x$ to $N$ in $B$: $\dist_{S}(x,s_i) <
    \epsilon\cdot \dist_{B}(x,N)$.
  \end{enumerate}  
\end{definition}

\begin{figure}
  \centering
  \subfigure[]{\includegraphics[scale=0.9]{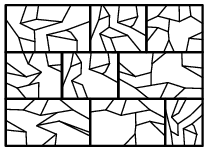}}\hfil\hfil
  \subfigure[]{\includegraphics[scale=0.9]{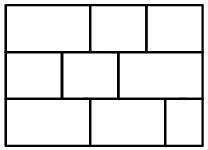}}\hfil\hfil
  \subfigure[]{\includegraphics[scale=0.9]{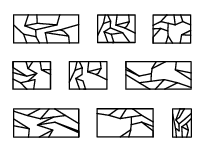}}\hfil\hfil
  \subfigure[]{\includegraphics[scale=0.9]{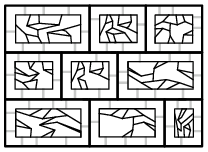}}\hfil\hfil
  \subfigure[]{\includegraphics[scale=0.9]{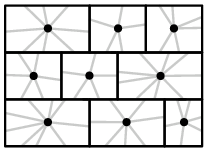}}
  \caption{(a)~An input graph $G$ with mortar graph $\MG$ given by bold
    edges in (b). (c)~The set of bricks corresponding to $\MG$ (d) A
    portal-connected graph, $\BC(\MG,\theta)$. The portal edges are
    grey.  (e)~$\BC(\MG,\theta)$ with the bricks contracted, resulting in
    $\BT(\MG,\theta)$.  The dark vertices are brick vertices. }
  \label{fig:mg} 
\end{figure} 

The mortar graph and the set of bricks are illustrated in
Figures~\ref{fig:mg}~(a),~(b) and~(c).  Constructing the mortar graph
for planar graphs first involves finding a 2-approximate Steiner tree
$T$~\cite{Mehlhorn88} and cutting open the graph along $T$ creating a new face $H$ and then: 

\begin{enumerate} 
\item Finding shortest paths between certain vertices of $H$.  These paths
  result in the $N$ and $S$ boundaries of the bricks.\label{mg:step-3}
\item Finding shortest paths between vertices of the paths found in
  Step~\ref{mg:step-3}.  These paths are called {\em columns}, do not
  cross each other, and have a natural order. \label{mg:step-4} 
\item
  Taking every $\kappa$th path found in Step~\ref{mg:step-4}.  These
  paths are called {\em supercolumns} and form the $E$ and $W$
  boundaries of the bricks. We sometimes refer to $\kappa$ as the
  \emph{spacing} of the supercolumns.
  \label{mg:step-5} 
\end{enumerate} 

The mortar graph is composed of the edges of $T$ (equivalently, $H$)
and the edges found in Steps~\ref{mg:step-3} and~\ref{mg:step-5}.
In~\cite{BorradaileKM09}, it is shown that the total length of the
mortar graph edges is at most $9\epsilon^{-1} \OPT$. For the purposes of this paper,
we bound the length of the mortar graph in terms of $\length{H}$.  The
following theorem can be easily deduced from~\cite{Klein06}
and~\cite{BorradaileKM09}:

\begin{theorem}[\cite{Klein06,BorradaileKM09}] \label{thm:planar-mg}
  Let $0 < \epsilon \leq 1$ and $G$ be a planar graph with outer face
  $H$ containing the terminals $Q$ and such that
  $\ell(H) \leq \alpha_0 \OPT$, for some constant
  $\alpha_0$. For $\alpha=(2\alpha_0+1)\epsilon^{-1}$, there is a
  mortar graph $\MG(G,Q,\epsilon)$ containing $H$ whose
  length is at most $\alpha \OPT$ and whose supercolumns have length
  at most $\epsilon \OPT$ with spacing $\kappa=\alpha_0
  \epsilon^{-2}(1+\epsilon^{-1})$.  The mortar graph can be
  found in $\Oof(n \log n)$ time.
\end{theorem} 

\subsection{A mortar graph for bounded-genus graphs: Overview}\label{sec:genus-mg}

We use Theorem~\ref{thm:planar-mg} to prove the existence of a mortar graph for
genus-$g$ embedded graphs. This section is devoted to
proving the following theorem: 

\begin{theorem}
  \label{thm:genus-mortar-graph} Let an embedded edge-weighted
  graph $G$ of Euler genus $g$, a subset of its vertices $Q$, an $0 <
  \epsilon \leq 1$, and $\mu \geq 1$ be given. For $\alpha = (32\mu g
  + 9)\epsilon^{-1}$, there is a mortar graph $\MG(G,Q,\epsilon)$ of
  $G$ such that the length of $\MG$ is $\leq \alpha \OPT$ and the
  supercolumns of $\MG$ have length $\leq \epsilon \OPT$ with spacing
  $\kappa = (16\mu g + 4)\epsilon^{-2}(1+\epsilon^{-1})$. The mortar
  graph can be found in $\Oof(n \log n)$ time.
\end{theorem}

Let $G_0=(V_0,E_0)$ be the input graph of genus $g_0$ and $Q$ be the
terminal set. In a first preprocessing step, we delete a number of
unnecessary vertices and edges of $G_0$ to obtain a graph $G=(V,E)$ of
genus $g \leq g_0$ that still contains every
$(1+\epsilon)$-approximate solution for terminal set $Q$ for all $0
\leq \epsilon \leq 1$ while fulfilling certain bounds on the length of
shortest paths. In the next step, we find a \emph{cut graph} $\CG$ of
$G$ that contains all terminals and whose length is bounded by a
constant times $\OPT$. We cut the graph open along $\CG$, so that it
becomes a planar graph with a simple cycle $\sigma$ as boundary, where
the length of $\sigma$ is twice that of $\CG$. See
Figure~\ref{fig:cutgraph} for an illustration. Afterwards, the
remaining steps of building the mortar graph can be the same as in the
planar case, by way of Theorem~\ref{thm:planar-mg}.

For an edge $e=vw$ in $G_0$, we let 
$$ \dist_{G_0}(r,e) = \min \set{\dist_{G_0}(r,v),\dist_{G_0}(r,w)} + \length{e} $$ 
and say that $e$
is \emph{at distance} $\dist_{G_0}(r,e)$ from $r$. If the root vertex
represents a contracted graph $H$, we use the same terminology with
respect to $H$.

\subsection{Preprocessing the input graph} \label{sec:preprocess}

Our first step is to apply the following preprocessing procedure:

\begin{center} \fbox{
    \begin{minipage}[h]{0.99\linewidth}
      \noindent\textbf{Algorithm $\preprocess(G_0,Q,\mu)$. } \\
      \begin{tabular}{ll}
        \textit{Input. } & an arbitrary graph $G_0$, terminals $Q \subseteq \vrt{G_0}$, a constant $\mu$\\
        \textit{Output. } & a preprocessed subgraph of $G_0$\\
      \end{tabular} \\[-2ex]
      \begin{enumerate} 
      \item Find a 2-approximate Steiner tree $T_0$ for $Q$ and contract it to a  vertex $r$.
      \item Find a shortest-path tree rooted at $r$.
      \item Delete all vertices $v$ and edges $e$ of $G_0$ with\\
        $\dist_{G_0}(r,v),\dist_{G_0}(r,e) > 2\mu\length{T_0}$.
      \end{enumerate}
    \end{minipage} }
\end{center}

\noindent Any deleted vertex or edge is at distance $> 2\mu
\length{T_0} > 2\mu \OPT$ from any terminal and hence can not be part
of a $(1+\epsilon)$-approximation for any $0 \leq \epsilon \leq 1$. We
call the resulting graph $G=(V,E)$ and henceforth use $G$ instead of
$G_0$ in our algorithm. The preprocessing step can be accomplished in
linear time: step~1 using M\"{u}ller-Hannemann and
Tazari's algorithm~\cite{TazariMuellerh09-DAM} and step~2 using Henzinger et al.'s
algorithm~\cite{HenzingerKRS97}. Trivially, we have
\begin{proposition}\label{prop:preprocess}
  All vertices and edges of $G$ are at distance at most $4\mu \OPT$ from $T_0$.
\end{proposition}

\subsection{Constructing a cut graph}

\begin{figure}
  \centering
  \includegraphics[height=3.5cm]{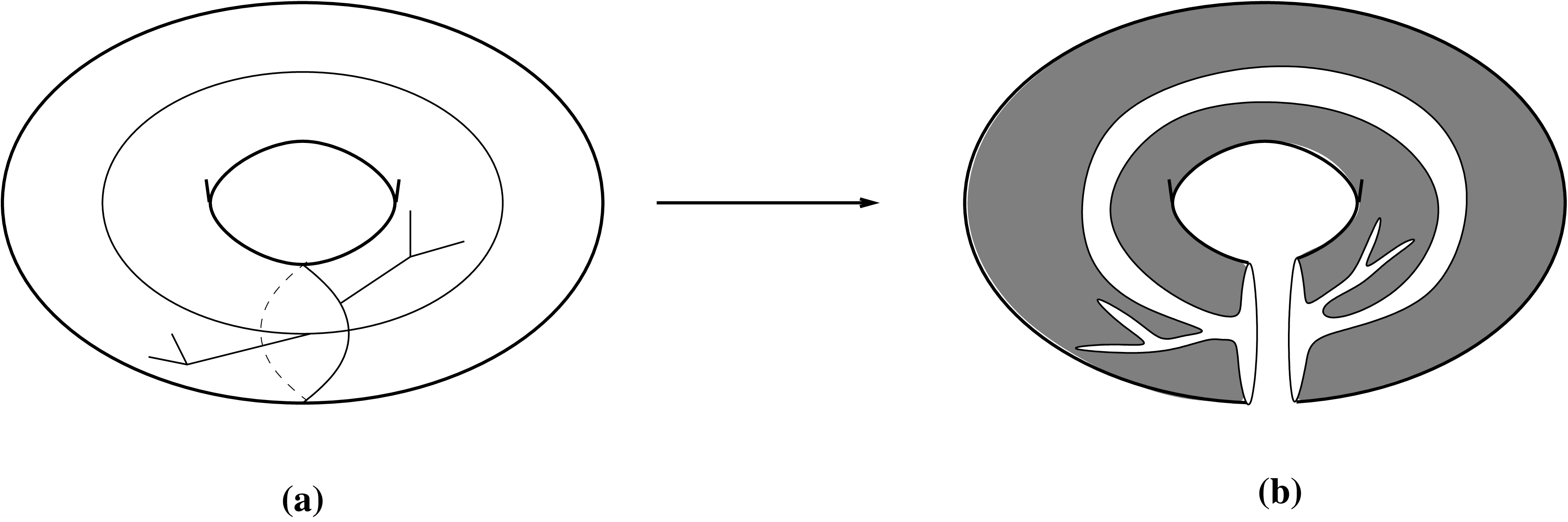}
  \caption{(a)~a cut graph of a tree drawn on a torus; (b) the result
    of cutting the surface open along the cut graph: the shaded area
    is homeomorphic to a disc and the white area is the additional
    face of the planarized surface.}
  \label{fig:cutgraph} 
\end{figure} 

A central fact that we use in this section and also in other parts of
our work is the following observation~\cite{EITTWY92}:

\begin{observation}\label{obs:planartree}
  Let $G$ be a planar graph and $T$ a spanning tree of $G$. Then the
  set of edges $\edgeG - \edgeT$ induces a spanning tree $T^\star$ in
  the dual $G^\star$. If $T$ is a minimum spanning tree of $G$, then
  $T^\star$ is a maximum spanning tree of $G^\star$.
\end{observation}

A similar lemma also holds for bounded-genus graphs: if $T$ is a
(minimum) spanning tree of $G$ and $T^\star$ a (maximum) spanning tree
of $G^\star - \edgeT$, then $T^\star$ is a (maximum) spanning tree of
$G^\star$ and the size of the set of remaining edges $X := \edgeG -
\edgeT -\edge{T^\star}$ is $g$, the Euler genus of $G$, by Euler's
formula. Eppstein~\cite{Eppstein03} defines such a triple
$(T,T^\star,X)$ as a \emph{tree-cotree decomposition} of $G$ and shows
that such a decomposition can be found in linear time for graphs on
both orientable and nonorientable surfaces.

In order to construct a cut graph, we start again with a
$2$-approximation $T_0$ and contract it to a vertex $r$. Next, we look
for a \emph{system of loops} rooted at $r$: iteratively find short
nonseparating cycles through $r$ and cut the graph open along each
cycle. Erickson and Whittlesey~\cite{EricksonWhittlesey05} showed
that, for orientable surfaces, taking the \emph{shortest} applicable
cycle at each step results in the shortest system of loops through
$r$. They suggest a linear-time algorithm using the
tree-cotree decomposition $(T,T^\star,X)$ of
Eppstein~\cite{Eppstein03}.  Eppstein showed:
\begin{lemma}[Lemma~2,~\cite{Eppstein03}]\label{lem:eppstein}
  Given a tree-cotree decomposition $(T,T^\star,X)$, the set of
  elementary cycles $\set{\mbox{loop}(T,e) : e \in X}$ is a cut graph
  of $G$ where $\mbox{loop}(T,e)$ is the closed walk formed by the
  paths in $T$ from $r$ to the endpoints of $e$ plus the edge~$e$.
\end{lemma}
 Eppstein's decomposition also works for nonorientable
embeddings. As we only need to bound the length (as opposed to
minimizing the length) of our cut graph, we
present a simpler algorithm below:

\begin{center} \fbox{
  \begin{minipage}[h]{0.99\linewidth}
      \noindent\textbf{Algorithm $\planarize(G_0,Q,\mu)$. } \\
      \begin{tabular}{ll}
        \textit{Input. } & a graph $G_0$ of genus $g$, terminals $Q \subseteq \vrt{G_0}$, a constant $\mu$\\
        \textit{Output. } & a preprocessed subgraph $G \subseteq G_0$ and a cutgraph $\CG$ of $G$\\
      \end{tabular} \\[-2ex]
      \begin{enumerate} 
      \item Apply $\preprocess(G_0,Q,\mu)$ and let $G$ be the obtained subgraph.
      \item Find a 2-approximate Steiner tree $T_0$ for $Q$ and contract it to a  vertex $r$.
      \item Find a shortest paths tree $\SPT$ rooted at $r$.
      \item Uncontract $r$ and set $T_1 = T_0 \cup \SPT$. {\em ($T_1$ is a
        spanning tree of $G$)}
      \item Find a spanning tree $T_1^\star$ in $G^\star -
        \edge{T_1}$. {\em ($T_1^\star$ is a spanning tree of $G^\star$)}
      \item Let $X := \edgeG - \edgeT -\edge{T^\star}$.
      \item Return $\CG := T_0 \cup \set{\mbox{loop}(T_1,e) : e \in X}$ together with $G$.
      \end{enumerate}
    \end{minipage} }
\end{center}

\begin{lemma} \label{lem:cut-graph} The algorithm $\planarize$
  returns a cut graph $\CG$ such that cutting
  $G$ open along $\CG$ results in a planar graph $G_p$ with a face
  $f_\sigma$ whose facial walk $\sigma$
\begin{enumerate} 
\item[(i)] is a simple cycle; 
\item[(ii)] contains all terminals (some terminals might appear more than
  once as multiple copies might be created during the cutting process); and
\item[(iii)] has length $\length{\sigma} \leq 2(8\mu g+2)\OPT$.
\end{enumerate}
The algorithm can be implemented in linear time.
\end{lemma}
\begin{proof}
  Clearly, $(T_1,T_1^\star,X)$ is tree-cotree decomposition of $G$ and
  so, by Lemma~\ref{lem:eppstein}, $\CG$ is a cut graph. By
  Euler's formula, we get that $|X| = g$, the Euler genus of $G$. 

  Each
  edge $e = vw \in X$ completes a (nonseparating,
  not necessarily simple) closed walk as follows: a shortest path $P_1$ from
  $T_0$ to $v$, the edge $e$, a shortest path $P_2$ from $w$ to $T_0$
  and possibly a path $P_3$ in $T_0$. By
  Proposition~\ref{prop:preprocess}, we know that $e$ is at distance
  at most $4\mu \OPT$ from $T_0$ and so, both $P_1$ and $P_2$, and at least one of
  $\{ P_1 \cup \set{e}, P_2 \cup \set{e}\}$ have length at most $4\mu
  \OPT$. Hence, we have that $\length{P_1 \cup \set{e} \cup P_2} \leq
  8\mu \OPT$. Because there are (exactly) $g$ such cycles in $\CG$, we get
  that
  $$
  \ell(\CG) \leq g \cdot 8\mu\OPT + \ell(T_0) \leq (8\mu g+2)\OPT .
  $$
  Since $\CG$ is a connected cut graph and $T^\star \cap \CG =
  \emptyset$, cutting $G$ open along $\CG$ results in a connected
  planar graph with boundary $\sigma$.  Each edge of $\CG$ appears
  twice in $\sigma$ and each edge of $\sigma$ is derived from $\CG$,
  so $\ell(\sigma) = 2 \ell(\CG)$ (see
  Fig.~\ref{fig:cutgraph}).


  As mentioned in the previous section, $T_0$ and $\SPT$ can be computed in
  linear time on bounded-genus
  graphs~\cite{HenzingerKRS97,TazariMuellerh09-DAM}. $T_1^\star$ can be obtained,
  for example, by a simple breadth-first-search in the dual. The remaining
  steps can also easily be implemented in linear time.\qed
\end{proof}

\subsection{Proof of Theorem~\ref{thm:genus-mortar-graph}} 

We complete the construction of a mortar graph for genus-$g$ embedded
graphs.

Let $G_p$ be the result of planarizing $G$ as guaranteed by
Lemma~\ref{lem:cut-graph}.  $G_p$ is a planar graph with
boundary $\sigma$ such that $\sigma$ spans $Q$ and has length $\leq
2(8\mu g+2)\OPT$.  Let $\MG$ be the mortar graph guaranteed by
Theorem~\ref{thm:planar-mg} as applied to $G$ with $\sigma$ as
its outer face.  Every edge of $\MG$ corresponds to an edge of $G$.
Let $\MG'$ be the subgraph of $G$ composed of edges corresponding to
$\MG$.  Every face $f$ of $\MG$ (other than $\sigma$) corresponds to a
face $f'$ of $\MG'$ and the interior of $f'$ is homeomorphic to a disk
on the surface in which $G$ is embedded.  It is easy to verify that
$\MG'$ is indeed a mortar graph of $G$; and the length bounds
specified in the statement of the theorem follow directly from
Theorem~\ref{thm:planar-mg} and the bound on the length of
$\sigma$. \qed

\subsection{Structure Theorem}\label{sec:structure-thm}

Along with the mortar graph, Borradaile et~al.~\cite{BorradaileKM09}
define an operation $\BC$ called {\em brick-copy} that allows a
succinct statement of the Structure Theorem.  For each brick $B$, a
subset of $\theta$ vertices are selected as {\em portals} such that
the distance along $\partial B$ between any vertex and the closest
portal is at most $\length{\partial B}/\theta$.  For every brick $B$,
embed $B$ in the corresponding face of $\MG$ and connect every portal
of $B$ to the corresponding vertex of $\MG$ with a zero-length {\em
  portal edge}: this defines $\BC(\MG,\theta)$.  $\BC(\MG,\theta)$ is
illustrated in Figure~\ref{fig:mg}~(d). We denote the set of all
portal edges by $\portaledges$. The following simple observation, proved
in~\cite{BorradaileKM09} holds also for bounded-genus graphs:

\begin{observation}[\cite{BorradaileKM09}] \label{lem:soln} If $A$ is a
  connected subgraph of $\BC(\MG, \theta)$, then $A -
  \portaledges$ is a connected subgraph of $G$ spanning the same vertices
  of $G$.
\end{observation} 

The following Structure Theorem is the heart of the correctness of the
$\PTAS$es.  

\begin{theorem}[Structure Theorem]
  Let $\PPP$ be one of the subset-connectivity problems \steiner,
  $\set{0,1,2}$-edge-connectivity \survive, or \subtsp. Let $G$ be
  an edge-weighted graph embedded on a surface, $Q \subseteq \vrtG$ a
  given set of terminals, and $0 < \epsilon \leq 1$. Let
  $\MG(G,Q,\epsilon)$ be a corresponding mortar graph of weight at
  most $\alpha \OPT$ and supercolumns of weight at most $\epsilon
  \OPT$ with spacing $\kappa$. There exist constants
  $\beta(\epsilon,\kappa)$ and $\theta(\alpha,\beta)$ depending
  polynomially on $\alpha$ and $\beta$ such that
  $$
  \OPT_\PPP(\BC(\MG,\theta),Q) \leq (1+c \epsilon)\OPT_\PPP(G,Q) \, ,
  $$ where $c$ is an absolute constant.  Here $\beta =
  o(\epsilon^{-2.5}\kappa)$ for \steiner and $\{0,1,2\}$-edge
  connectivity \survive and $\beta = \Oof(\kappa)$ for \subtsp.
  (Recall that $\alpha$ and $\kappa$ depend polynomially on
  $\epsilon^{-1}$ and $g$ by Theorem~\ref{thm:genus-mortar-graph}.)
\end{theorem}

It is due to our special way of defining and constructing a mortar
graph for bounded-genus graphs that this theorem follows immediately
as for the planar cases: the crucial point here is that our bricks are
always planar -- even when the given graph is embedded in a surface of
higher genus. The Structure Theorem for \steiner is proved
in~\cite{BorradaileKM09}, the case of $\set{0,1,2}$-edge-connectivity
\survive is studied in~\cite{BorradaileKlein08}, and we show that the
theorem holds for \subtsp in Section~\ref{sec:subtsp}. Note
  that for \subtsp, it is possible to obtain a singly exponential
  algorithm by following the spanner construction of
  Klein~\cite{Klein06} after performing the planarizing step
  (Lemma~\ref{lem:cut-graph}). Our presentation here is chosen to
  unify the methods for all problems studied.

The Structure Theorem essentially says that there is a constant
$\theta$ depending polynomially on $\epsilon^{-1}$ such that in
finding a near-optimal solution to $G$, we can restrict our attention
to $\BC(\MG,\theta)$. Whenever we wish to apply our framework to a
new problem, it is essential to prove a similar structure theorem for
the considered problem.

\section{Obtaining $\PTAS$es for bounded-genus graphs}

We present two methods of obtaining polynomial-time approximation
schemes.  The first is a generalization of the framework of
Klein~\cite{Klein06} for planar graphs that is based on finding a
\emph{spanner} for a problem, a subgraph containing a nearly optimal
solution having length $\Oof(\OPT)$.  In Section~\ref{sec:spanner} we
show how to find such a spanner and in Section~\ref{sect:ptas_spanner}
we generalize Klein's framework to higher genus graphs using the
techniques of Demaine et~al.~\cite{DemaineHM07}.  In the second
method, dynamic programming is done over the bricks of the mortar
graph.  This generalizes the framework of Borradaile
et~al.~\cite{BorradaileKM09} for planar graphs to higher genus
graphs.  While both methods result in $\Oof(n \log n)$ algorithms, the
first method is doubly exponential in a polynomial in $g$ and
$\epsilon^{-1}$ and the second is singly exponential.

\subsection{Spanner for Subset-Connectivity Problems}\label{sec:spanner}

A spanner is a subgraph of length
$\Oof_{\epsilon,g}(\OPT)$ that contains a $(1+\epsilon)$-approximate
solution.  Here we show how to find a spanner for bounded-genus graphs
and the subset-connectivity problems considered in this paper. After a
mortar graph is computed, the construction is, in fact, exactly the
same as in the planar cases, namely:
\begin{quote}
  For each brick $B$ defined by $\MG$ and for each subset $X$ of the
  portals of $B$, find the optimal Steiner tree of $X$ in $B$ (using
  the method of Erickson et~al.~\cite{EricksonMV87}).  The spanner $\Gspan$ is
  the union of all these trees over all bricks plus the edges of the
  mortar graph.
\end{quote}
To prove the correctness of our spanner theorem for the case of
$\set{0,1,2}$-edge-connectivity \survive, we need to appeal to the
following result of Borradaile and Klein, which we have simplified the
statement of here:

\begin{theorem}[{\cite[Theorem~5]{BorradaileKlein08}}]\label{thm:2ec-trees}
  Consider an instance of the $\{0,1,2\}$-edge connectivity problem.
  There is a feasible solution $S$ to this instance that is a subgraph of
  $\BC(MG)$ such that 
  \begin{itemize}
  \item $\ell(S) \leq (1+c\epsilon)\OPT$ where $c$ is an absolute
    constant, and
  \item the intersection of $S$ with any brick $B$ is a set of $O(1)$
    trees the set of leaves of which are portals.
  \end{itemize}
\end{theorem}


\begin{theorem}[Spanner Theorem]\label{thm:spanner}
  Let $G$ be an edge-weighted graph embedded on a surface of Euler
  genus $g$ and $Q \subseteq \vrtG$ a given set of terminals. There
  exists a spanner $\Gspan \subseteq G$ such that
  \begin{description} 
  \item[$\Gspan$ is spanning:] $\Gspan$ contains a $(1+c\epsilon)$-approximate solution to
    \steiner, $\set{0,1,2}$-edge-connected \survive, and \subtsp; and
  \item[$\Gspan$ is short:] $\length{\Gspan} \leq f(\epsilon,g) \OPT$;
  \end{description}
  where the function $f(\epsilon,g)$ is singly exponential in a
  polynomial in $\epsilon^{-1}$ and $g$, and $c$ is an absolute
  constant. The spanner can be found in $\Oof(n \log n)$ time.
\end{theorem}

\begin{proof}
  Given a mortar graph $\MG(G,Q,\epsilon)$ as guaranteed by
  Theorem~\ref{thm:genus-mortar-graph}, a spanner is constructed
  as specified above.  As in~\cite{BorradaileKM09}, the time to find
  $\Gspan$ is $\Oof(n \log n)$.  It was proved in~\cite{BorradaileKM09}
  that $\length{\Gspan} \leq (1+2^{\theta+1})\length{\MG}$.
  Therefore, $\length{\Gspan} \leq (1+2^{\theta+1})\alpha \OPT$ and
  $f(\epsilon,g) = (1+2^{\theta+1}) \alpha$ (recall that $\alpha$ and
  $\theta$ depend polynomially on $\epsilon^{-1}$ and $g$). 
  
  Now we show that $\Gspan$ contains a near-optimal solution to each
  problem.
  For \steiner, the proof follows directly from the Structure Theorem:
  the intersection of a minimal solution in $\BC(\MG,\theta)$ with a
  brick $B$ is a forest whose leaves are portals.

  For $\{0,1,2\}$-edge-connected \survive, we appeal to
  Theorem~\ref{thm:2ec-trees}: By the Structure Theorem, there
  is a solution $H$ in $\BC(\MG)$ that has length at most
  $(1+c\epsilon)\OPT$.  For each brick $B$, let $H_B$ be the
  intersection of $H$ with $B$.  
  $H_B$ is the union of trees.
  Replace each tree with the Steiner
  tree spanning the same subset as found in the spanner construction.
  Let $H'$ be the graph resulting from all such replacements:
  $\length{H'} \leq \length{H} \leq (1+c\epsilon)\OPT$.  By
  Observation~\ref{lem:soln}, the edges of $H' - \portaledges$ induce
  a solution to the problem of length at most $(1+c\epsilon)\OPT$.

  For \subtsp, the proof is similar.  By the Structure Theorem, there
  is a tour $T$ of the terminals $Q$ in $\BC(\MG)$ that has length at
  most $(1+c\epsilon)\OPT$.  For each brick $B$, let $K$ be a
  connected component of the intersection of $T$ with $B$.  Because
  the terminals are in $\MG$ and not in $B$, $K$ is a path between
  portals of $B$: replace $K$ with the Steiner tree (i.e.\ a shortest
  path) connecting these two portals found in the construction of the
  spanner\footnote{Note that to construct a spanner for \subtsp, we
    need only shortest paths between pairs of portals.}.  Let
  $T'$ be the tour resulting from all these replacements: $\length{T'}
  \leq \length{T} \leq (1+c\epsilon)\OPT$.  Appealing to
  Observation~\ref{lem:soln}, the edges of $T' - \portaledges$ induce
  a solution of length at most $(1+c\epsilon)\OPT$.\qed
\end{proof}

\subsection{\PTAS via Spanner}\label{sect:ptas_spanner}

In order to apply the \PTAS framework of Klein~\cite{Klein08} to
bounded-genus graphs, we need the following Contraction Decomposition
Theorem due to Demaine et al.:

\begin{theorem}[{\cite[Theorem~1.1]{DemaineHM07}}] \label{thm:genus-baker}
  For a fixed genus $g$, and any integer $\eta \geq 2$ and for every
  graph $G$ of Euler genus at most $g$, the edges of $G$ can be
  partitioned into $\eta$ sets such that contracting any one of the sets
  results in a graph of treewidth at most $\Oof(g^2 \cdot \eta)$.  Furthermore,
  such a partition can be found in $\Oof(g^{5/2}n^{3/2} \log n)$ time.
\end{theorem}
Recent techniques~\cite{CabelloChambers07} for finding shortest
noncontractible cycles of embedded graphs have improved the above
running time to $\Oof(n \log n)$.\footnote{We would like to thank Jeff
  Erickson for pointing out in private communication that the
  algorithm given in~\cite{CabelloChambers07} works for both
  orientable and nonorientable surfaces.}

We review the four steps of the framework in our setting:
\begin{description}
\item[1. Spanner Step:] Find a spanner $\Gspan$ of $G$ according to
  Theorem~\ref{thm:spanner}.
\item[2. Thinning Step:] For $\eta = f(\epsilon,g)/\epsilon$ (where
  $f(\epsilon,g)$ is the function given in Theorem~\ref{thm:spanner}),
  let $S_1, \ldots, S_\eta$ be the partition of the edges of $\Gspan$ as
  guaranteed by Theorem~\ref{thm:genus-baker}.  Let $S^*$ be the set
  in the partition with minimum weight: $\length{S^*} \leq \epsilon
  \OPT$.  Let $\Gthin$ be the graph obtained from $\Gspan$ by contracting the
  edges of $S^*$.  By Theorem~\ref{thm:genus-baker}, $\Gthin$ has
  treewidth at most $\Oof(g^2\epsilon^{-1}f(\epsilon,g))$.
\item[3. Dynamic Programming Step:] Use dynamic programming (see,
  e.g.~\cite{KorachSolel90}) to find the optimal solution to the
  problem in $\Gthin$.
\item[4. Lifting Step:] Convert this solution to a solution in $G$ by
  incorporating some of the edges of $S^*$.  For \steiner, at most one
  copy of each edge of $S^*$ is introduced to maintain
  connectivity~\cite{BorradaileKM09}. In the case of $\{0,1,2\}$-edge connected
  \survive, at most two copies of each edge of $S^*$ are
  required~\cite{BorradaileKlein08}. For \subtsp, the method was
  explained in~\cite{Klein06}.
\end{description}

\paragraph{Analysis of the running time.}
By Theorem~\ref{thm:spanner}, the spanner step takes
$\Oof_{\epsilon,g}(n \log n)$ time (with singly exponential dependence
on polynomials in $g$ and $\epsilon^{-1}$).  By
Theorem~\ref{thm:genus-baker}, thinning takes time $\Oof(n \log
n)$ using~\cite{CabelloChambers07}. Dynamic programming takes time
$2^{\Oof(g^2\epsilon^{-1}f(\epsilon,g))}n$: because $f(\epsilon,g)$ is
singly exponential in polynomials in $g$ and $\epsilon^{-1}$, this
step is doubly exponential in polynomials in $g$ and $\epsilon^{-1}$.
Lifting takes linear time.  Hence, the overall running time is
$\Oof(2^{\Oof(g^2\epsilon^{-1}f(\epsilon,g))}n + n \log n)$.

\subsection{\PTAS via Dynamic Programming over the Bricks} \label{sec:dp}

In~\cite{BorradaileKM09}, Borradaile et al.\ present a \PTAS that is
singly exponential in a polynomial in $\epsilon^{-1}$ for \steiner in
planar graphs. The idea is to incorporate the spanner step into the
dynamic programming step and to use a somewhat modified thinning
step. To this end, the operator \emph{brick-contraction} $\BT$ is
defined to be the application of the operation $\BC$ followed by
contracting each brick to become a single vertex of degree at most
$\theta$ (see Figure~\ref{fig:mg}(e)). The thinning algorithm
decomposes the mortar graph $\MG$ into parts  of
\emph{bounded dual radius} (implying bounded treewidth). Applying
$\BT$ to each part maintains bounded dual radius. The algorithm computes
optimal Steiner trees inside the bricks using the method
of~\cite{EricksonMV87} only at the leaves of the dynamic programming tree, thus
eliminating the need of an a-priori constructed spanner. The
interaction between subproblems of the dynamic programming is
restricted to the portals, of which there are few.

For embedded graphs with genus $> 0$, the concept of bounded dual
radius does not apply in the same way.  We deal with treewidth
directly and obtain the following algorithm:
we apply the Contraction Decomposition
Theorem~\ref{thm:genus-baker}~\cite{DemaineHM07} to $\BT(\MG)$ and
contract a set of edges $S^\star$ in $\BT(\MG)$. However, we apply a
special weight to portal edges so as to prevent them from being
included in $S^\star$. Also, in $\BT(\MG)$, we slightly modify the
definition of contraction: after contracting an edge, we do not delete
parallel portal edges.  Because portal edges connect the mortar graph
to the bricks, they are not parallel in the graph in which we find a
solution via dynamic programming. The details are given below.

\begin{center} \fbox{
    \begin{minipage}[h]{0.95\linewidth}
      \noindent\textbf{Algorithm $\thinning(G,\MG)$. } \\
      \begin{tabular}{ll}
        \textit{Input. } & a graph $G$ of fixed genus $g$, a mortar graph $\MG$ of $G$\\
        \textit{Output. } & a set $S^\star \subseteq \edge{\BT(\MG)}$,\\
                          & a tree decomposition $(T,\chi)$ of $\BT(MG)/S^\star$\\
      \end{tabular} \\[-2ex]
      \begin{enumerate} 
      \item Assign weight \length{\partial F} to each portal edge
        in a face $F$ of $\BT(\MG)$.
      \item Apply the Contraction Decomposition
        Theorem~\ref{thm:genus-baker} to $\BT(\MG)$ with $\eta := 3\theta \alpha
        \epsilon^{-1}$ to obtain edge sets $S_1,\dots,S_\eta$; let $S^\star$ be
        the set of minimum weight.
      \item If $S^\star$ includes a portal edge $e$ of a brick $B$
        enclosed in a face $F$ of $\MG$,\\ add $\partial F$ to
        $S^\star$ and mark $B$ as ignored.
      \item Let $\MGthin := \BT(\MG) / S^{\star}$ (but do not delete parallel portal edges).
      \item Let $(T, \chi)$ be a tree decomposition of width $\Oof(g^2 \cdot \eta)$ of $\MGthin$.
      \item For each vertex $b$ of $\MGthin$ that represents an unignored
        contracted brick with portals $\set{p_1,\dots,p_{\theta}}$:
        \begin{enumerate} 
        \item[6.1.] Replace every occurrence of $b$ in $\chi$
          with $\set{p_1,\dots,p_{\theta}}$;
        \item[6.2.] Add a bag $\set{b,p_1,\dots,p_{\theta}}$ to
          $\chi$ \\ and connect it to a bag containing $\set{p_1,\dots,p_{\theta}}$.
        \end{enumerate}
      \item Reset the weight of the portal edges back to zero.
      \item Return $(T,\chi)$ and $S^\star$.
      \end{enumerate}
    \end{minipage} }
\end{center}

\begin{lemma}\label{lem:thinning}
The algorithm $\thinning(G,\MG)$ returns a set of edges $S^\star$ and a tree
decomposition $(T,\chi)$ of $\BT(\MG)/S^\star$, so that
\begin{itemize} 
\item[(i)] the treewidth of $(T,\chi)$ is at most $\xi$ where
  $\xi(\epsilon,g) = \Oof(g^2\eta\theta) =
  \Oof(g^3\epsilon^{-2}\theta^2)$;\\ in particular, $\xi$ is polynomial
  in $\epsilon^{-1}$ and $g$;
\item[(ii)] every brick is either
  \begin{itemize}
  \item marked as ignored, or 
  \item none of its portal edges are in $S^\star$; and
  \end{itemize}
\item[(iii)] $\length{S^\star} \leq \epsilon \OPT$.
\end{itemize}
\end{lemma}

\begin{proof}
  We first verify that $(T,\chi)$ is indeed a tree decomposition. For a vertex
  $v$ and a tree decomposition $(T',\chi')$, let $T'_v$ denote the subtree of
  $T'$ that contains $v$ in all of its bags. Let us denote the tree
  decomposition of step (5) by $(T^0,\chi^0)$. For each brick vertex $b$ and
  each of its portals $p_i$, we know that $T^0_b$ is connected and $T^0_{p_i}$
  is connected and that these two subtrees intersect; it follows that after
  the replacement in step (6.2), we have that $T_{p_i} = T^0_b \cup T^0_{p_i}$
  is a connected subtree of $T$ and hence, $(T,\chi)$ is a correct tree
  decomposition. Note that Theorem~\ref{thm:genus-baker} guarantees a tree
  decomposition of width $\Oof(g^2\eta)$ if any of $S_1,\dots,S_\eta$ are contracted;
  and in step (3), we only add to the set of edges to be contracted. Hence,
  the treewidth of $(T^0,\chi^0)$ is indeed $\Oof(g^2\eta)$ and with the
  construction in step (6.1), the size of each bag will be multiplied by a
  factor of at most $\theta$. This shows the correctness of claim (i). The
  correctness of claim (ii) is immediate from the construction in step (3). It
  remains to verify claim (iii).
  
  Let $L$ denote the weight of $\BT(\MG)$ after setting the weights of the
  portal edges according to step (1) of the algorithm. We have that
  \begin{align*}
  L &\leq \length{\MG} + \sum_F \length{\partial F}\theta \leq \alpha
  \OPT + \theta \sum_F \length{ \partial F}\\
  &\leq \alpha \OPT + \theta
  \cdot 2\alpha \OPT \leq 3 \theta \alpha \OPT \, .
  \end{align*}
  \noindent Hence, the weight of $S^\star$, as selected in step (2),
  is at most $L / \eta \leq \frac{3 \theta \alpha \OPT}{3 \theta \alpha
    \epsilon^{-1}} \leq \epsilon \OPT$. The operation in step
  (3) does not add to the weight of $S^\star$: if $\partial F$ is
  added to $S^\star$, the additional weight is subtracted when the
  corresponding portal-edge weights are set to zero in step (7).\qed
\end{proof}

If a brick is ``ignored'' by $\thinning$, the boundary of
its enclosing mortar graph face is completely added to $S^\star$.
Because $S^\star$ can be added to the final solution, every potential
connection through that brick can be rerouted through $S^\star$ around
the boundary of the brick. The interior of the brick is not needed.

An almost standard dynamic programming algorithm for bounded-treewidth
graphs (cf.~\cite{ArnborgProskurowski89,KorachSolel90}) can be applied
to $\Gthin$ and $(T,\chi)$.  However, for the leaves of the tree
decomposition that are added in step (6.2) of the $\thinning$
procedure, the cost of a subset of portal edges is calculated as,
e.g., the cost of the minimum Steiner tree interconnecting these
portals in the corresponding brick. Because the bricks are planar,
this cost can be calculated by the algorithm of~\cite{EricksonMV87}
for Steiner tree or~\cite{BorradaileKlein08} for $2$-edge
connectivity.  Because all the portal edges of this brick are present
in this bag (recall that we do not delete parallel portal edges after
contractions), all possible solutions restricted to the corresponding
brick will be considered.  Because the contracted brick vertices only
appear in leaves of the dynamic programming tree, the rest of the
dynamic programming algorithm can be carried out as in the standard
case.

\paragraph{\bf Analysis of the running time.} As was shown for the
planar Steiner tree PTAS~\cite{BorradaileKM09}, the total time spent
in the leaves of the dynamic programming is $\Oof(4^\theta n)$. The
rest of the dynamic programming takes time $\Oof(2^{\Oof(\xi)}
n)$. The running time of the thinning algorithm is dominated by the
Contraction Decomposition Theorem~\ref{thm:genus-baker} which is
$\Oof_g(n \log n)$~\cite{CabelloChambers07}. Hence, the total time is
$\Oof(2^{\Oof(\xi)} n + 2^{\poly(g)}n \log n)$ for the general case;
in particular, this is singly exponential in $\epsilon^{-1}$ and $g$,
as desired. This proves Theorem~\ref{thm:main_ptas}.

\section{A Structure Theorem for Subset TSP}
\label{sec:subtsp}

Here we prove the Structure Theorem for \subtsp.
While this theorem can be used to obtain a
\PTAS for the subset tour problem in planar graphs, a \PTAS for this
problem~\cite{Klein06} predates the mortar graph framework.

Like \steiner and \survive, the Structure Theorem
(Section~\ref{sec:structure-thm}) must be proved
(Section~\ref{sec:struct-TSP}) for the \subtsp problem.  To this
end, in this section, we state and prove a local structure theorem
(Theorem~\ref{thm:tsp-prop}, Figure~\ref{fig:tour}).  This local
structure theorem describes how to replace the intersection of a tour
with a brick to reduce the number of times the tour crosses the
boundary of the brick: each crossing contributes to the size of the
dynamic-programming table.  While the intersection of a tour with a
brick is quite simple (a set of brick boundary-to-boundary paths), in
modifying the tour we must be careful to maintain that our solution is
still a tour.

We will use the following lemma due to Arora:
\begin{lemma}[Patching Lemma~\cite{Arora03}] \label{lem:patching}
  Let $S$ be any line segment of length $s$ and $\pi$ be a closed path
  that crosses $S$ at least thrice.  Then we can break the path in all
  but two of these places and add to it line segments lying on $S$ of
  total length at most $3s$ such that $\pi$ changes into a closed path
  $\pi'$ that crosses $S$ at most twice.
\end{lemma}
This lemma applies to embedded graphs as well.  Note: the patching
process connects paths in the tour that end on a common side of $S$ by
a subpath of $S$.

\begin{theorem}[TSP Property of Bricks] \label{thm:tsp-prop} Let $B$
  be a brick of graph $G$ with boundary $N \cup E \cup S \cup W$
  (where $E$ and $W$ are supercolumns).  Let $T$ be a tour in $G$ such
  that $T$ crosses $E$ and $W$ at most $2$ times each.  Let $H$ be the
  intersection of $T$ with $B$.  Then there is another subgraph of
  $B$, $H'$, such that:
  \begin{enumerate}[(T1)] 
  \item $H'$ has at most $\beta(\epsilon)$ joining vertices with
    $\partial B$.
  \item $\length{H'} \leq (1+5 \epsilon) \length{H}$.
  \item There is a tour in the edge set $T \setminus H \cup H'$ that
    spans the vertices in $\partial B \cap T$.
  \end{enumerate}
  In the above, $\beta(\epsilon) = \Oof(\kappa)$.
\end{theorem}

\begin{proof}
  Let $H$ be the subgraph of $T$ that is strictly enclosed by $B$
  (i.e., $H$ contains no edges of $\partial B$).  We can write $H$ as
  the union of 3 sets of minimal $\partial B$-to-$\partial B$ paths
  ${\cal P}_{S \vee N} \cup {\cal P}_{E \vee W} \cup {\cal P}_{S \wedge N}$
  where paths in: ${\cal P}_{S \vee N}$ either start and end on $S$ or
  start and end on $N$; ${\cal P}_{E \vee W}$ start on $E$ or $W$ and end on
  $\partial B$; ${\cal P}_{S \wedge N}$ start on $S$ and end on
  $N$. For the constructions below, refer to Figure~\ref{fig:tour}.

  \begin{figure}[h]
    \centering 
    \subfigure[]{\includegraphics[scale=0.35]{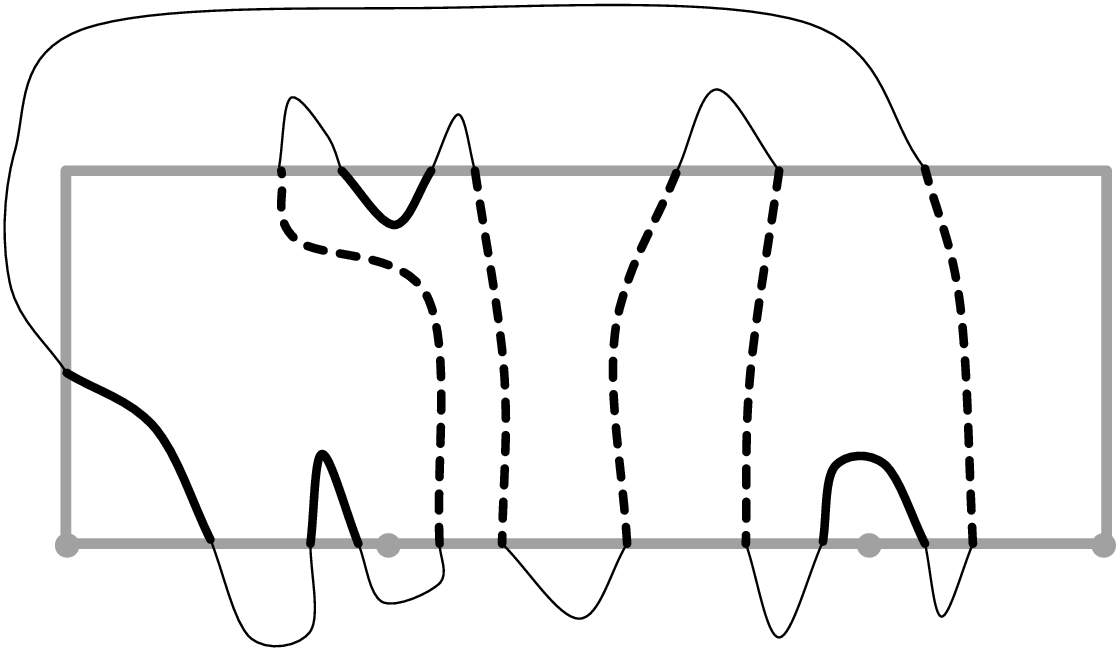}}
    \subfigure[]{\includegraphics[scale=0.35]{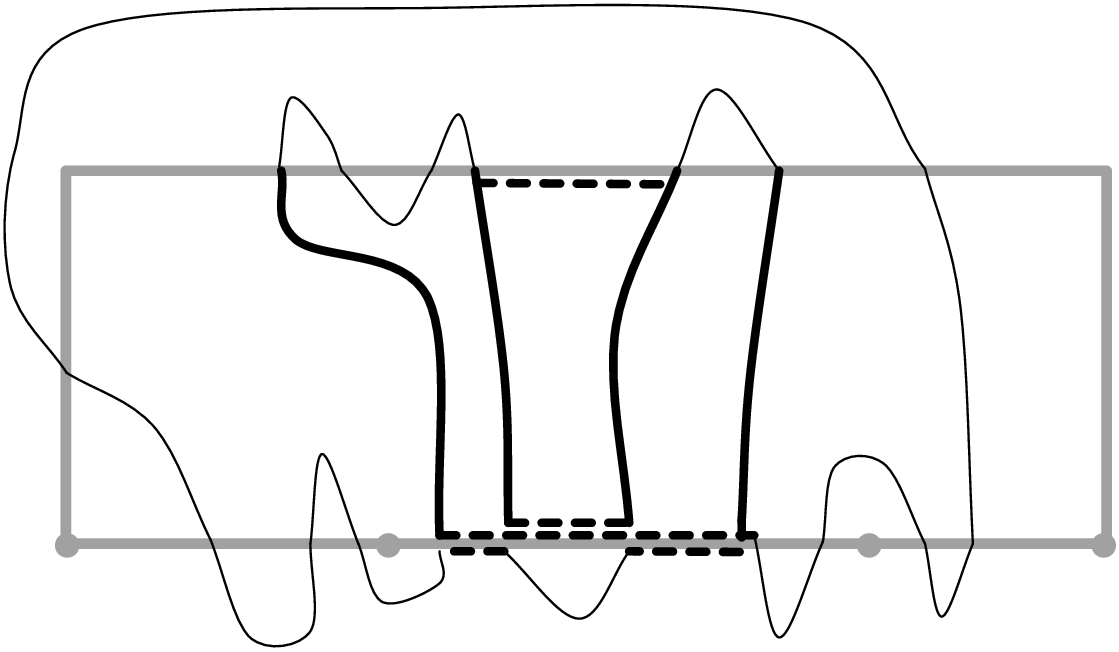}}
    \subfigure[]{\includegraphics[scale=0.35]{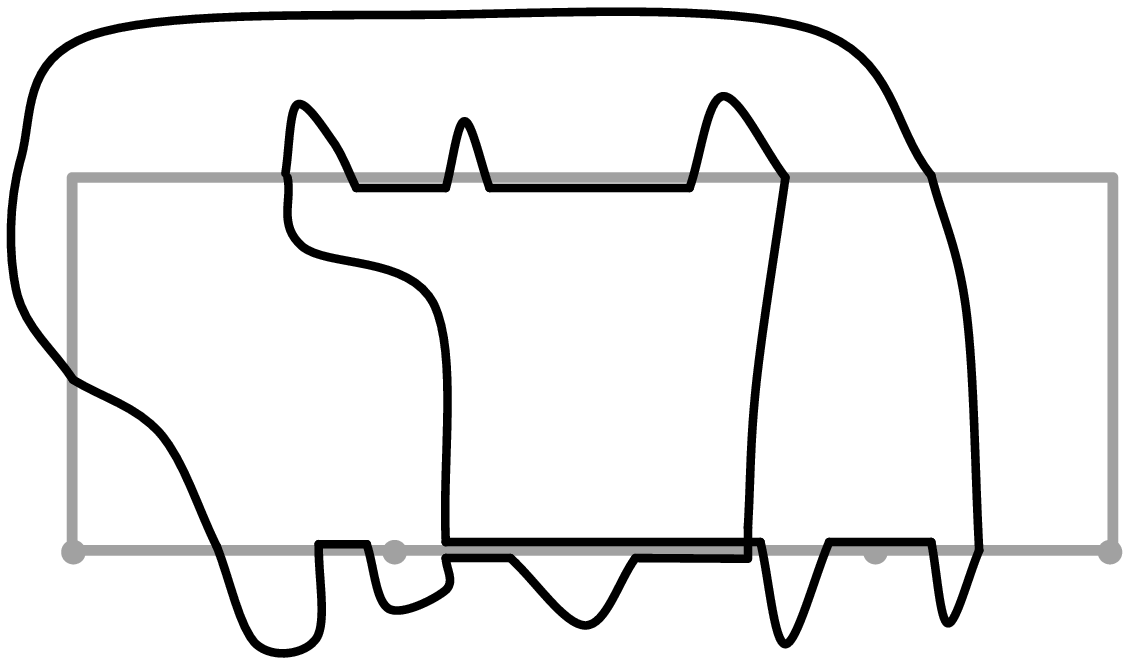}}
    \caption{(a)~A brick with a tour crossing through it.  The bold
      paths are in $H$.  The bold vertices are $s_0, s_1, s_2, \ldots,
      s_k$. The dotted paths are in ${\cal P}_{S \wedge N}$, the first four of
      which are in ${\cal X}_1$.  (b)~The patching process introduces
      the dotted paths on the lower boundary $S$ of the brick and
      reroutes the tour to cross $S$ twice between $s_1$ and $s_2$.
      The dotted subpath $L$ of the top boundary $N$ of the brick is
      used to replace the portion of the tour between its endpoints.
      (c)~The tour after the entire construction given by
      Theorem~\ref{thm:tsp-prop}.}
    \label{fig:tour}
  \end{figure}

  Since $T$ crosses $E$ and $W$ at most 4 times, $|{\cal P}_{E \vee W}| \leq
  4$.  Therefore, this set of paths result in at most 8 joining
  vertices with $\partial B$.

  For each path $P \in {\cal P}_{S \vee N}$, let $\widehat P$ be the minimal
  subpath of $\partial B$ that spans $P$'s endpoints.  Let $\widehat {\cal
    P}_{S\vee N}$ be the resulting set of paths.  As $N$ is 0-short
  and $S$ is $\epsilon$-short, we have
  \begin{equation}
    \length{\widehat {\cal P}_{S\vee N}} \leq
    (1+\epsilon)\length{{\cal P}_{S\vee N}}. \label{eq:S}
  \end{equation}
  Since $\widehat {\cal P}_{S \vee N}$ are subpaths of $\partial B$,
  they have no joining vertices with $\partial B$.  Since paths in
  $\widehat {\cal P}_{S \vee N}$ correspond one-to-one with paths in
  ${\cal P}_{S \vee N}$, $T \setminus {\cal P}_{S \vee N} \cup
  \widehat {\cal P}_{S \vee N}$ is a tour.  See
  Figure~\ref{fig:tour}(a).

  It remains to deal with the paths in ${\cal P}_{S \wedge N}$.  Let
  $s_0, s_1, s_2, \ldots, s_k$ (where $k \leq \kappa$) be the vertices
  of $S$ guaranteed by the properties of the bricks (see
  Definition~\ref{def:mg}).  Let ${\cal X}_i$ be the subset of paths
  of ${\cal P}_{S\wedge N}$ that start on $S[s_i,s_{i+1})$, i.e.\ the
  vertices between $s_i$ and $s_{i+1}$ including $s_i$ but not
  $s_{i+1}$.

  If $|{\cal X}_i| > 2$, we do the following: Let $P_i$ be the path in
  ${\cal X}_i$ whose endpoint $x$ on $S$ is closest to $s_{i+1}$.  Let
  $Q_i$ be the subpath of $S$ from $s_i$ to $x$.  By the properties of
  the bricks, $\length{Q_i} \leq \epsilon \length{P_i}$.  Apply the
  Patching Lemma to the tour $T$ and path $Q_i$; the new tour, $T'$,
  crosses $Q_i$ at most twice.  However $T'$ may still have many
  joining vertices with $Q_i$.  Let ${\cal Q}_i$ be the subpaths of
  $Q_i$ that are added to the tour.

  Let ${\cal L}_i$ be a maximal set of $N$-to-$N$ paths in ${\cal X}_i
  \cup {\cal Q}_i$. ${\cal L}_i$ accounts for all but two
  (corresponding to the two crossings of $T'$ with $Q_i$) of the
  joining vertices of $T'$ with $Q_i$.  For each path $L \in {\cal
    L}_i$, let $\widehat L$ be the minimal subpath of $N$ that spans
  $L$'s endpoints and let $\widehat {\cal L}_i$ be the resulting set
  of paths.  Replacing ${\cal L}_i$ with $\widehat {\cal L}_i$ is
  still a tour, since the paths have a one-to-one correspondence.
  However, the resulting tour may no longer span all terminals on $Q_i$.
  Adding in two copies of $Q_i$ fixes this.  Since $N$ is 0-short,
  $\length{\widehat {\cal L}_i} \leq \length{{\cal L}_i}$.

  Let $\widehat {\cal X}_i = {\cal X}_i \cup {\cal Q}_i \setminus
  {\cal L}_i \cup \widehat {\cal L}_i \cup Q_i \cup Q_i$.  Replacing
  ${\cal X}_i$ with $\widehat {\cal X}_i$ is still a tour, as argued
  above.  Since the additional length added is at most 5 copies of
  $Q_i$, we have:
  \begin{equation}
    \length{\widehat {\cal X}_i} \leq \length{{\cal X}_i} + 5 \length{Q_i} \leq \length{{\cal
        X}_i} + 5\epsilon \length{P_i} \leq (1+5\epsilon) \length{{\cal X}_i} \label{eq:Q}
  \end{equation}
  Since ${\cal L}_i$ accounted for all but 2 of the joining vertices
  of $T'$ with $Q_i$ - so all but 4 of the joining vertices of ${\cal
    X}_i$ with $\partial B$, and $\widehat {\cal L}_i$ has no joining
  vertices with $\partial B$, $\widehat {\cal X}_i$ has at most 4
  joining vertices with $\partial B$.

  Let $\widehat {\cal P}_{S \wedge N} = 
  \bigcup_i \widehat {\cal X}_i$.
  $\widehat {\cal P}_{S \wedge N}$
  has at most $6\kappa$ joining vertices with $\partial B$ and,
  by Equation~\eqref{eq:Q},
  \begin{equation}
    \length{\widehat {\cal P}_{S\wedge N}} \leq (1+5\epsilon)\length{{\cal P}_{S \wedge N}}.\label{eq:NS}
  \end{equation}
  
  Let $\widehat H$ be the union of the paths in ${\cal P}_{E \vee W}$,
  $\widehat {\cal P}_{S\vee N}$ and $\widehat {\cal P}_{S\wedge N}$.
  Combining Equations~(\ref{eq:S}) and~(\ref{eq:NS}), we find that
  $\length{\widehat H} \leq (1+5\epsilon) \length{H}$.  By
  construction, the edges in $T \setminus H \cup \widehat H$ contain
  a tour.  $\widehat H$ has $6\kappa + 8$ joining
  vertices with $\partial B$.\qed
\end{proof}

\subsection{Proof of the Structure Theorem for \subtsp} \label{sec:struct-TSP}

Using the TSP Property of Bricks, we prove the Structure Theorem
(Section~\ref{sec:structure-thm}) for \subtsp.

Let $T$ be the optimal tour spanning terminals $Q$ in $G$.  From $T$
we build a tour $\widehat T$ spanning $Q$ in $\BC(MG)$ such that
$\length{\widehat T} \leq (1+c \epsilon) \length{T}.$

Let $C$ be a supercolumn.  By the Patching Lemma, if $T$ crosses $C$
at least thrice, we can add to $T$ at most three copies of $C$ and
create a new tour that crosses $C$ at most twice.  Let $T_1$ be the
tour that results from applying the Patching Lemma to each
supercolumn.  Because the sum of the weights of the supercolumns is at
most $\epsilon \OPT$,
\begin{equation}
  \length{T_1} \leq (1+3\epsilon) \length {T}.\label{eq:s1}
\end{equation}

Let $B$ be a brick of $G$.  Let $H$ be the intersection of $T_1$ with
$B$.  By the construction above, $T_1$ satisfies the requirements of
Theorem~\ref{thm:tsp-prop}: let $H'$ be the guaranteed subgraph of
$B$.  We replace $H$ with $H'$ in $T_1$.  Let $T_2$ be the tour
resulting from such replacements over all the bricks.
Theorem~\ref{thm:tsp-prop} guarantees that 
\begin{equation}
  \length{T_2} \leq (1+ 5 \epsilon) \length{T_1}.\label{eq:s2}
\end{equation}

Now we map the edges of $T_2$ to a subgraph of $\BC(MG)$.  Every edge
of $G$ has at least one corresponding edge in $\BC(MG)$.  For every
edge $e$ of $T_2$, we select one corresponding edge in $\BC(MG)$ as
follows: if $e$ is an edge of $MG$ select the corresponding mortar
edge of $\BC(MB)$, otherwise select the unique edge corresponding to
$e$ in $\BC(MG)$.  (An edge of $\BC(MG)$ is a {\em mortar edge} if it
is in $MG$.)  This process constructs a subgraph $T_3$ of $\BC(MG)$
such that
\begin{equation}\label{eq:s3}
  \length{T_3} = \length{T_2}.
\end{equation}

Because $T_3$ is not connected, we connect it via portal and mortar
edges.  Let $V_B$ be the set of joining vertices of $T_3$ with
$\partial B$ for a brick $B$ of $\BC(MG)$.  For any vertex $v$ on the
interior boundary $\partial B$ of a brick, let $p_v$ be the portal on
$\partial B$ that is closest to $v$, let $P_v$ be the shortest
$v$-to-$p_v$ path along $\partial B$ and let $P_v'$ be the
corresponding path of mortar edges.  Let $e$ be the portal edge
corresponding to $p_v$.  Add $P_v$, $P_v'$, and $e$ to $T_3$.  Repeat
this process for every $v \in V_B$ and for every brick $B$, building a
graph $\widehat T$. This completes the definition of $\widehat T$.
From the construction, $\widehat T$ is a tour spanning the terminals
$Q$ in $\BC(MG)$.

Now we analyze the increase in length:
\begin{equation}\label{eq:s4}
  \length{\widehat T}\leq \length{T_3}+ \sum_{B \in {\cal B}} \sum_{v
    \in V_B} (\length{P_v} + \length{e} + \length{P'_v}),
\end{equation}
and we have:
\begin{eqnarray*}
  \sum_{B \in {\cal B}} \sum_{v \in V_B} \length{P_v} +
  \length{e} + \length{P'_v} 
  &=& 2 \sum_{B \in {\cal B}} \sum_{v \in V_B} \length{P_v}\mbox{,
    because }\ell({\mbox{portal edges}})=0 \\
  &\leq& 2 \sum_{B \in {\cal B}} \sum_{v \in V_B} \length{\partial
    B}/\theta\mbox{, by the choice of portals} \\
  &\leq& 2 \sum_{B \in {\cal B}}
  \frac{\beta}{\theta} \length{\partial B}\mbox{, by
    Theorem~\ref{thm:tsp-prop}}\\
  &\leq& 2 \frac{\beta}{\theta} 2\alpha (\epsilon^{-1},g)
  \OPT\mbox{, by Theorem~\ref{thm:genus-mortar-graph}}\\
  &\leq& \epsilon\, \OPT\mbox{, for $\theta = 4\epsilon^{-1}\beta
    \alpha$, as required.}
\end{eqnarray*}

Combining Equations~(\ref{eq:s1}),~(\ref{eq:s2}),~(\ref{eq:s3})
and~(\ref{eq:s4}), we obtain $\length{\widehat T} \leq
(1+3\epsilon)(1+5\epsilon)\length{T} + \epsilon \OPT \leq (1+c \epsilon) \OPT$.
The Structure Theorem is proved for the \subtsp problem.

\section{Conclusion and Outlook}

We presented a framework to obtain $\PTAS$es on bounded-genus
graphs for subset-connectivity problems, where we are given a graph
and a set of terminals and require a certain connectivity among the
terminals. Specifically, we obtained the first $\PTAS$ for \steiner on
bounded-genus graphs running in $\Oof(n\log n)$-time with a constant
that is singly exponential in $\epsilon^{-1}$ and the genus of the
graph. Our method is based on the framework of Borradaile et
al.~\cite{BorradaileKM09} for planar graphs; in fact, we generalize
their work in the sense that basically any problem that is shown to
admit a $\PTAS$ on planar graphs using their framework easily
generalizes to bounded-genus graphs using the methods presented in
this work. In particular, this gives rise to $\PTAS$es in
bounded-genus graphs for \subtsp (Section~\ref{sec:subtsp}),
$\set{0,1,2}$-edge-connected \survive~\cite{BorradaileKlein08}, and
also \steinerforest~\cite{BateniHM10}.

A natural question is to ask what other classes of graphs admit a
\PTAS for the problems discussed in this work. An important
generalization of bounded-genus graphs are proper classes of graphs
that are closed under taking \emph{minors}. Such \emph{$H$-minor-free
  graphs} have earned much attention in recent years. Many hard
optimization problems have been shown to admit $\PTAS$es and
fixed-parameter algorithms on these classes of graphs; see,
e.g.,~\cite{DemaineHajiaghayi05-PTAS,Grohe03}. But subset-connectivity
problems, specifically \subtsp and \steiner, remain important open
problems~\cite{Grohe03,DemaineHM07}. Both a spanner theorem and a
contraction decomposition theorem are still missing for the
$H$-minor-free case. Very often, results on $H$-minor-free graphs are
first shown for planar graphs, then extended to bounded-genus graphs,
and finally obtained for $H$-minor-free graphs. This is due to the
powerful decomposition theorem of Robertson and
Seymour~\cite{GraphMinors16} that essentially says that every
$H$-minor-free graph can be decomposed into a number of parts that are
``almost embeddable'' in a bounded-genus surface. We conjecture that
our framework extends to $H$-minor-free graphs via this decomposition
theorem. The advantage of our methodology is that handling weighted
graphs and subset-type problems are naturally incorporated,
and thus it might be possible to combine all the steps for a
potential \PTAS into a single framework for $H$-minor-free graphs
based on what we presented in this work.  Hence, whereas our
work is an important step towards this generalization, still a number
of hard challenges remain; see also~\cite{DemaineHM07} for a further
discussion on this matter.




\end{document}